\documentclass[pra,aps,superscriptaddress,twocolumn,nofootinbib,longbibliography,unpublished]{revtex4-2}
\pdfoutput=1
\usepackage{amsfonts,amsmath,amssymb,amsthm,bm,bbm,centernot,color,xcolor,enumerate,enumitem,graphicx,hyperref,mathdots,mathtools,multirow,soul,subfigure,thmtools}
\usepackage[capitalise, noabbrev]{cleveref}
\hypersetup{colorlinks=true,linkcolor=blue,citecolor=blue,urlcolor=blue}
\usepackage[margin=2cm]{geometry}
\usepackage[utf8]{inputenc}
\usepackage[sort&compress]{natbib}
\usepackage[mathscr]{eucal}

\newcommand{\ket}[1]{|#1\rangle}

\newcommand{\proj}[1]{|#1\rangle\!\langle#1|}
\newcommand{\Id}{{\mathbb I}}
\newcommand{\Tr}{{\textrm {Tr} }}

\newcommand{\diag}{{\mathrm {diag}}}

\newcommand\scalemath[2]{\scalebox{#1}{\mbox{\ensuremath{\displaystyle #2}}}}
\def\E{ {\cal E} }
\def\C{ {\cal C} }
\def\D{ {\cal D} }
\def\M{ \mathscr{M}}
\def\S{ {\cal S}}
\def\I{ {\cal I}}
\def\L{ {\cal L}}
\def\Q{ {\cal Q}}
\def\U{ {\cal U}}
\def\V{ {\cal V}}
\def\W{ {\cal W}}
\def\P{ {\cal P}}
\def\MM{{{\mathscr M}_{\!s, \hat{\boldsymbol{n}}}}}
\def\PP{\P_{\!\scalemath{0.6}{\mathscr M}_{\!s, \hat{\boldsymbol{n}}}}}
\def\UU{\U_{\scalemath{0.6}{\mathscr M}_{\!s, \hat{\boldsymbol{n}}}}}
\def\CC{\C_{\scalemath{0.6}{\mathscr M}_{\!s, \hat{\boldsymbol{n}}}}}
\def\DD{\D_{\E_{u}}}

\theoremstyle{plain}
\newtheorem{theorem}{Theorem}
\newtheorem{lemma}[theorem]{Lemma}

\newtheorem{remark}[theorem]{Remark}
\theoremstyle{definition}
\newtheorem{definition}{Definition}

\definecolor{darkgreen}{RGB}{0,175,0}
\definecolor{ppblue}{RGB}{46,117,182}
\definecolor{ppred}{RGB}{197, 90, 17}

\begin{document}

\title{Measurement sharpness and disturbance tradeoff}

\author{Nayere Saberian}
\email{n.saberian.qi@gmail.com}
\affiliation{Department of Physics, Ferdowsi University of Mashhad, Mashhad, Iran}
\author{Seyed Javad Akhtarshenas}
\email{akhtarshenas@um.ac.ir}
\thanks{Corresponding author}
\affiliation{Department of Physics, Ferdowsi University of Mashhad, Mashhad, Iran}
\author{Fereshte Shahbeigi}
\email{fereshte.shahbeigi@uj.edu.pl}
\affiliation{Faculty of Physics, Astronomy and Applied Computer Science, Jagiellonian University, 30-348 Krakow, Poland}

\begin{abstract}
Obtaining information from a quantum system through a measurement typically disturbs its state. The postmeasurement states for a given measurement, however, are not unique and highly rely on the chosen measurement model, complicating the puzzle of information-disturbance.  Two distinct questions are then in order. Firstly, what is the minimum disturbance a measurement may induce? Secondly, when a fixed disturbance occurs, how informative is the possible measurement in the best-case scenario? Here, we propose various approaches to tackle these questions and provide explicit solutions for the set of unbiased binary qubit measurements and postmeasurement state spaces that are equivalent to the image of a unital qubit channel. In particular, we show there are different tradeoff relations between the sharpness of this measurement and the average fidelity of the premeasurement and postmeasurement state spaces as well as the sharpness and quantum resources preserved in the postmeasurement states in terms of coherence and discord-like correlation once the measurement is applied locally.
\end{abstract}

\maketitle

\section{Introduction}\label{Sec-Introduction}

The no-information-without-disturbance theorem is a significant result of quantum mechanics, which asserts that acquiring information from quantum systems is necessarily at the cost of disturbing their states. Extensive research has been done into the tradeoff between the obtained information and consequent disturbance in a system \cite{asher-info-dist,banaszek-inf-dist,banaszek-inf-dist-2,Daariano-info-dist,maccone-info-dist,sciarrino-info-dist,sacchi-inf-dist,buccemi-info-dist,werner-info-dist,horodecki-info-dist,Busch2009,chribella-inf-dist,yong-info-dist,hiroaki-info-dist,kumar-info-dist}.

In order to retrieve information, one has to measure the system. Mathematically speaking, a quantum measurement is a map from the set of quantum states to the classical ones, namely probability distributions. This definition generally leaves the system under the measurement with various potential postmeasurement states depending on the adopted measurement strategy \cite{Busch2016}, see Fig.~\ref{fig:indirect-measurement} for an approach called indirect measurement \cite{ozawa}. The non-uniqueness of the postmeasurement states, however, introduces additional difficulties in studying the disturbance within the state space caused by measurements.

In a series of papers \cite{Heinosaari_2016,Taiko2013,TeikoFP2014,Teiko2018}, authors proposed  the seminal concept of channel-measurement compatibility to address the non-uniqueness problem. This approach provides a framework to consider both; all possible postmeasurement states for a given measurement and all possible measurements that may induce the same disturbance, i.e., the same postmeasurement states.

Roughly speaking, any method employed to measure a quantum measurement $\M$ leads to a specific overall disturbance in the state space that is equivalent to applying a quantum channel $\E_\M$ on the states.  Henceforth, such a channel is referred to as being compatible with the measurement \cite{Heinosaari_2016,Taiko2013,TeikoFP2014,Teiko2018}.  The set $\C_\M=\{\E_\M\}$ of all channels  compatible with a given measurement $\M$, therefore, encapsulates all the possible disturbances that measuring $\M$ may induce.

Conversely, the set $\D_\E=\{\M_\E\}$ of all possible compatible measurements with  a channel $\E$ indicates what measurements can be implemented,  and therefore, how much information can be extracted, at the cost of the same disturbance in the states. In fact, any quantum channel $\E$ admits different decompositions in terms of some completely positive and trace non-increasing operations $\Phi_i$, i.e., $\E=\sum\Phi_i$ for different sets of $\{\Phi_i\}$. Different Kraus representations are an example of this notion. Each of such sets yields a set of possible output states with some associated probabilities. These probabilities can be assigned to a quantum measurement $\M_\E$, which we called compatible with the quantum channel $\E$. Trivially, different decompositions can give rise to different measurements $\M_\E$.

It is worth mentioning that while this approach provides a structural framework, that works independently of any specific quantification of disturbance-information relation \cite{Taiko2013,TeikoFP2014,Teiko2018}, the question of characterising the sets $\C_\M$ for a given $\M$ and $\D_\E$ for a given $\E$ is generally open. The necessary and sufficient condition for compatibility of quantum measurements and quantum channels is known only for binary and unbiased qubit measurements and unital qubit channels \cite{Teiko2018}.

Nevertheless, a general consensus is that the most informative measurements are included in the set of projective ones, i.e., sharp measurements. Thus, adding noise to these measurements makes them unsharp and less informative with the possibility of inducing less disturbance in the outputs. There are different results on the postmeasurement disturbance, however to our best knowledge, there is no result quantifying such a relation between  sharpness and disturbance.

Here, we address this question by applying the notion of measurement-channel compatibility in the set of unbiased binary qubit measurements and unital qubit channels. For this purpose, we first present, explicitly, the known generalisation of channel-measurement compatibility from the set of Pauli channels to the set of unital qubit ones mentioned in Ref. \cite{Teiko2018}. Next, for the aforementioned sets of channels and measurements, we will show there are some tradeoff relations between sharpness of  the measurement and quantum properties of the postmeasurement states. These properties are considered in three different ways; (i) their average fidelity with premeasurement states, (ii) their coherence, and (iii) their discord-like correlations. The first one quantifies the proximity between pre- and postmeasurement states, whereas, the latter two characterise the extent to which quantum features are preserved in states after being measured.

We show the sharper the measurement is, the more quantum properties are lost in the postmeasurement states, thus, more disturbance is induced. Our results confirm that in the extreme case of a projective measurement, which is completely sharp, all the possible postmeasurement states, independent of the realisation method of the measurement, are completely classical. On the other extreme, when the measurement effects are proportional to identity, thus no information can be extracted from the systems, one may adopt a measuring method that preserves the quantum properties of all states after the measurement.

The paper is organised as follows. We bring in Section \ref{sec:preliminaries}  the preliminary information and mathematical background for the problem we address. The main results are presented in Section~\ref{section:results}, where we first state an explicit approach for the known generalisation of the channel-measurement compatibility to the set of unital qubit channels and then use this  generalisation to investigate the sharpness-disturbance tradeoffs in two separate subsections. The paper is concluded in Section~\ref{sec:Conclusion}, and the proof of some of the results is presented in two appendices.

\begin{figure}
    \centering
    \includegraphics[width=\columnwidth]{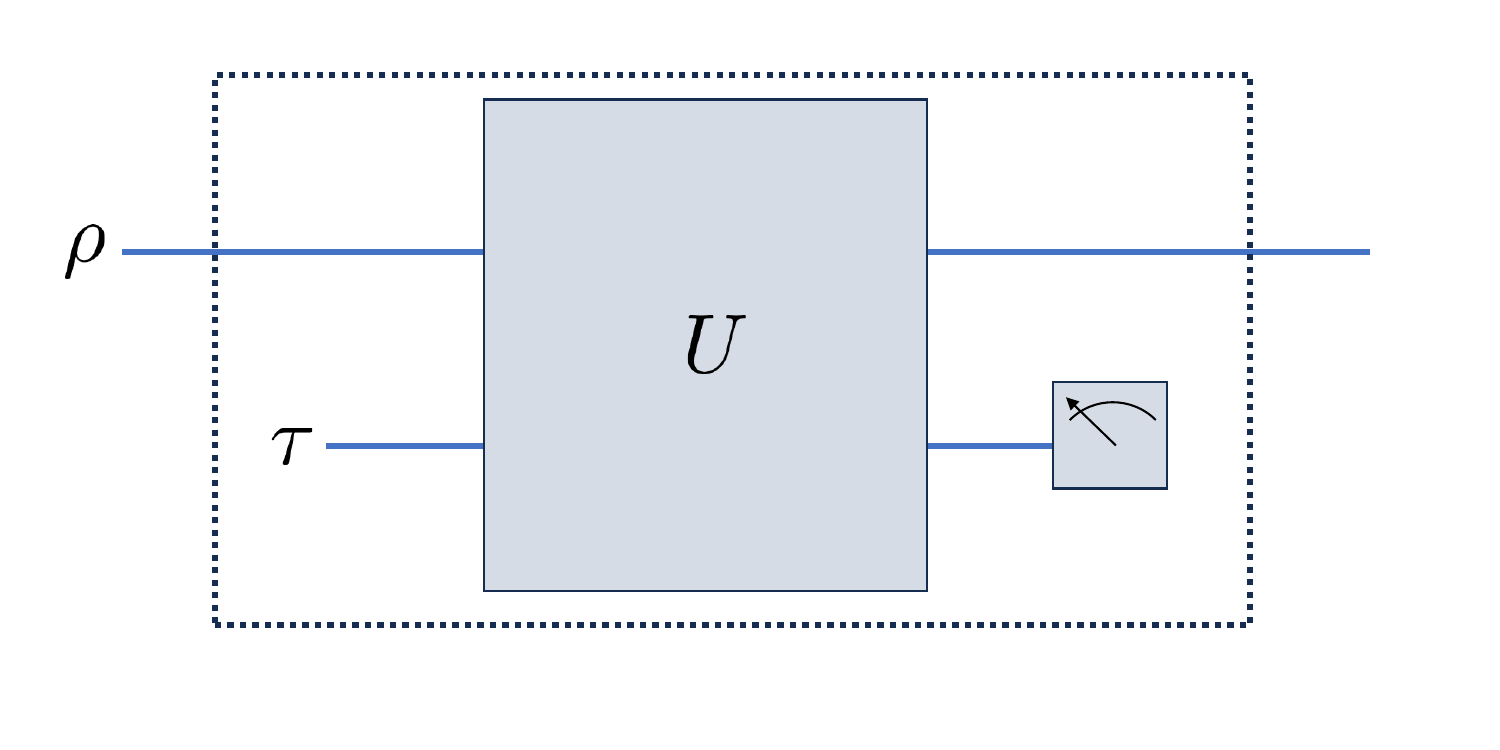}
    \caption{To apply a measurement $\M$ on state $\rho$, one can apply a joint unitary evolution on $\rho$ and some ancillary system, then measure the ancillary system by some pointer observable \cite{ozawa}, see also Naimark's dilation theorem. This setup contains non-uniqueness in the postmeasurement states. For example, both unitary operators $U_1$ equal to the swap operator and $U_2=\proj0\otimes\Id+\proj1\otimes\sigma_x$ for $\tau=\proj0$ model the same projective measurement $\M=\{\proj0,\proj1\}$. Then, the postmeasurement state space in the first scenario is independent of $\rho$ and includes only $\proj0$, while in the latter is the diameter of the Bloch ball on the $z$-axis.}
    \label{fig:indirect-measurement}
\end{figure}

\section{Preliminaries}
\label{sec:preliminaries}
In the following three subsections, we review the elementary definitions and bring the mathematical background  on quantum measurements, quantum channels, and their compatibility.
\subsection{Quantum measurements}
\label{subsection:measurements}
An $N$-outcome quantum measurement $\M$, also known as a positive operator-valued measurement (POVM) and an observable in some literature, is a set of positive semidefinite operators satisfying the complementary relation, i.e.,  $\M=\{M_i\}_{i=0}^{N-1}$ where $M_i\geq0$ and $\sum M_i = \Id$. A quantum measurement provides a map from the set of quantum states to classical ones, meaning that measuring $\M$ on a quantum state $\rho$ results in   the probability distribution $p_i=\Tr[\rho M_i]$ for  the outcome $i$. The positivity of the measurement elements implies the non-negativity of probabilities $p_i$  and the  complementary relation guarantees that they sum  up to one.   A special POVM, satisfying the orthogonality condition $M_i M_j = \delta_{ij} M_i$ for all $i,j$,  is called a  projective-valued measurement, or a projective measurement for brevity.

Two characteristics for a given POVM are primarily considered in this paper, namely biasedness and sharpness \cite{BuschFP2009}. A quantum measurement $\M$ is called unbiased if $\Tr{M_i}=d/N$ for all outcomes $i$.  The sharpness, on the other hand, defines how closely  $\M$ resembles a projective measurement. Different measures for biasedness and sharpness (sometimes unsharpness) were introduced \cite{MassarPRA2007, BuschFP2009, BaekSR2016, LiuPRA2021, MitraIJTF2022}.

Here, we are mainly concerned with unbiased binary (two-outcome) qubit measurements. Such a POVM is denoted by  $\MM=\{M_+,M_-\}$ where its two effects are parametrised by
\begin{eqnarray}\label{eqn: POVM elements}
M_{\pm}=\frac{1}{2}(\Id\pm s \hat{\boldsymbol{n}}\cdot \boldsymbol{\sigma}),
\end{eqnarray}
in which $s \in [0,1]$  carries the sharpness information and $\hat{\boldsymbol{n}} \in \mathbb{R}^3$, with $\|\hat{\boldsymbol{n}}\|=1$,  represents its direction.
Being  positive operators of unit trace,  $M_{\pm}$  correspond to a pair of points inside the Bloch ball, laying symmetrically around the origin   on a diameter defined by $\hat{\boldsymbol{n}}$.

Trivially, $\MM$ for $s=1$ is a projective measurement and $s=0$ gives a maximally noisy one. Generally, the sharpness of  $\MM$ can be quantified by $s$ or some monotone functions of $s$ such as $s^2$ which is in agreement with the unsharpness quantifiers defined in  \cite{LiuPRA2021, MitraIJTF2022}.  In Appendix \ref{sharpness-measures} we see that for the unbiased binary qubit measurements,  $1-s^2$ is equal to the unsharpness defined in  \cite{LiuPRA2021, MitraIJTF2022}.

\subsection{Quantum channels}
\label{subsection:channels}
A quantum channel is a completely positive and trace-preserving map acting on quantum states.
A Pauli channel $\E_{\vec{p}}$ is a special example of a  qubit channel which is defined by
\begin{eqnarray}
\label{eq:pauli-channel}
    \E_{\vec{p}}(\rho)=\sum_{j=0}^{3} p_j \sigma_j \rho \sigma_j,
\end{eqnarray}
where $\sigma_0$ is the identity matrix and $\sigma_j$ for $j\in\{1,2,3\}$ is a Pauli matrix. Also, the probability vector $\vec{p}\in \mathbb{R}^4$ satisfies  $0\leq p_i \leq 1$ and $\sum_{i=0}^{3} p_i=1$. Hence, a regular tetrahedron in $\mathbb{R}^3$ \cite{ZyczkowskiBook2017} can represent the set of Pauli channels. The set of Pauli channels, denoted by $\P$, is a measure zero subset of the unital ones, denoted by $\U$, which map identity to itself.

Three different features of quantum channels are considered in this paper; (i) input-output fidelity \cite{karimipourPRA2020,Shahbeigi2021}, (ii) quantumness \cite{ShahbeigiPRA2018},  and (iii) local quantum uncertainty \cite{GirolamiPRL2013} of the Choi-Jamio{\l}kowski state \cite{JAMIOLKOWSKI,CHOI}. In what follows we  briefly review these properties.

{\it (i) Input-output fidelity:}
The input-output fidelity is a measure of the performance of a quantum channel obtained by taking the fidelity between input and their corresponding output states after evolving by $\E$ averaged over  the Haar measure of pure states
\begin{align}
\label{eq:fidelity}
    \overline{F}_{\E}=\int\langle\psi|\E(\proj\psi)|\psi\rangle \mathrm{d}\psi.
\end{align}
This is equal to one only for the identity channel that does nothing on the state space. Generally, for $\E$ acting on $d$-dimensional states, this quantity respects the bounds $\frac{1}{d+1}\leq \overline{F}_{\E}\leq1$. Moreover, one may improve  this fidelity for a channel $\E$ by concatenating it with some other channels. In that case, the \emph{corrected fidelity},  $\overline{\mathscr{F}}_\E$, is given by \cite{karimipourPRA2020,Shahbeigi2021}
\begin{align}
\label{eq:corrected-fidelity}
     \overline{\mathscr{F}}_\E&:=\max_{\E'}\overline{F}_{(\E'\circ\E)}= \max_{\E''}\overline{F}_{(\E\circ\E'')}\\
     &= \max_{\E',\E''}\overline{F}_{(\E'\circ\E\circ\E'')}.\nonumber
\end{align}
For a Pauli channel $\E_{\vec{p}}$, these fidelities are \cite{karimipourPRA2020,Shahbeigi2021}
\begin{align}
    \overline{F}_{\E_{\vec{p}}}&=\frac13(1+2p_0),\label{eq:fidelity-pauli}\\
     \overline{\mathscr{F}}_{\E_{\vec{p}}}&=\overline{F}_{(\W_m\circ\E_{\vec{p}})}=\frac13(1+2p_m),\label{eq:corrected-fidelity-pauli}
\end{align}
where $p_m=\max\{p_0,p_1,p_2,p_3\}$ and $\W_m(\rho)=\sigma_m\rho\sigma_m$ is the  unitary Pauli channel corresponding to $p_m$.


{\it (ii) Quantumness of quantum channels:}
 Quantumness of  a channel $\E$  is defined  as the  average quantum coherence of the state space after the channel acts on, minimised over all orthonormal basis sets \cite{ShahbeigiPRA2018}
 \begin{align}
 \label{eq:quantumness}
     \Q_\E=N_C\min_{\{\ket i\}}\int C\!\left(\E(\rho)\right)\mathrm{d}\mu(\rho),
 \end{align}
 where $C$ is a proper coherence measure and $N_C$ is its associated normalisation constant chosen such that for the identity map the quantumness achieves unity.  Trivially, $\Q_\E$ quantifies the coherence preserved in the states undergoing a  channel. It is restricted to $0\leq\Q_\E\leq1$. The upper and lower bounds are respectively satisfied with equality only by a unitary channel and  a quantum-classical (q-c) one \cite{Holevo_qc} which  completely erases the quantum coherence of all states \cite{ShahbeigiPRA2018}. Moreover, it was shown \cite{ShahbeigiPRA2018} that for any channel $\E$
 \begin{equation}
 \label{eq:quantumness-invariance}
 \Q_\E=\Q_{(\V_2\circ\E\circ\V_1)},
 \end{equation}
 where $\V_1$ and $\V_2$ are  any pair of unitary channels.

Using the squared $l_1$-norm of coherence \cite{coherencePRL} as the coherence measure, the  quantumness of the Pauli channels then reads \cite{ShahbeigiPRA2018}
\begin{eqnarray}
\label{Q-PauliChannels}
    \mathcal{Q}_{\E_{\vec{p}}}=(p_0^{\downarrow}-p_1^\downarrow)^2+(p_2^\downarrow-p_3^\downarrow)^2,
\end{eqnarray}
where  $\vec{p}^\downarrow$ is obtained by rearranging $\vec{p}$  components in the non-increasing order.

{\it (iii) Local quantum uncertainty (LQU) of the Choi-Jamio{\l}kowski state:} Another approach to factor the quantum features of a channel in is to consider the quantum correlation it leaves in a bipartite state when acts locally on it.  A natural choice for such a bipartite state is a maximally entangled one. Thus, the quantum correlation of the Choi-Jamio{\l}kowski state associated with   a quantum channel reveals the quantum features of the channel. Using the notion of local quantum uncertainty \cite{GirolamiPRL2013}, as a measure of nonclassical correlation, we measure the quantum correlation of the Choi-Jamio{\l}kowski state associated with a given channel.

The LQU  is defined as the minimum skew information \cite{WignerYanasePNAS1963} achievable on a single local measurement and   has a closed form with respect to subsystem $A$ of a  qubit-qudit state $\varrho_{AB}$ as \cite{GirolamiPRL2013}
\begin{eqnarray}\label{LQU-A}
    \L_{\varrho_{AB}|A}=1-\lambda_{\max}\{W_{AB}\}.
\end{eqnarray}
Above   $\lambda_{\max}$ denotes the maximum eigenvalue of the $3\times 3$ symmetric matrix $W_{AB}$  whose entries for $i,j\in\{1,2,3\}$ are defined by
\begin{eqnarray}
    (W_{AB})_{ij}=\Tr\{\sqrt{\varrho_{AB}} (\sigma_{i}\otimes \mathbb{I}) \sqrt{\varrho_{AB}} (\sigma_{j}\otimes \mathbb{I})\}.
\end{eqnarray}

It is straightforward to show that LQU of the Choi-Jamio{\l}kowski state $\varrho_\E$ of a  quantum channel $\E$ satisfies $0\leq\L_{\varrho_\E|A}\leq 1$ with equality on lower-bound only for q-c channels and on upper-bound only for the unitary ones. Also, this measure is invariant under unitary transformations of the map, i.e., for all channels $\E$, and unitary maps   $\V_1$ and $\V_2$
\begin{equation}
\label{eq:lqu-invariance}
    \L_{\varrho_\E|A}=\L_{\varrho_{(\V_2\circ\E\circ\V_1)}|A}.
\end{equation}

For a Pauli channel $\E_{\vec{p}}$, the Choi-Jamio{\l}kowski state is given by
\begin{align}
\varrho_{\E_{\vec{p}}}:&=(\E_{\vec{p}}\otimes\I)\proj{\psi_+}\\ \nonumber
       &=\frac{1}{2}\begin{pmatrix}
        p_0+p_3 & 0 & 0 & p_0-p_3 \\
        0 & p_1+p_2 & p_1-p_2 & 0 \\
        0 & p_1-p_2 & p_1+p_2 & 0 \\
         p_0-p_3 & 0 & 0 & p_0+p_3
    \end{pmatrix},
\end{align}
where $\I$ is the identity channel and $\ket{\psi_+}=(\ket{00}+\ket{11})/\sqrt{2}$ is a maximally entangled state.  The Choi-Jamio{\l}kowski state $\varrho_{\E_{\vec{p}}}$ is therefore  a Bell-diagonal state for which $W_{AB}=\diag\{P_1,P_2,P_3\}$ where
\begin{eqnarray}\label{Pi}
{P}_i=2(\sqrt{p_0 p_i}+\sqrt{p_j p_k}),
\end{eqnarray}
and  $(i,j,k)$ are different choices of $(1,2,3)$. The LQU of $\varrho_{\E_{\vec{p}}}$ then reads
\begin{eqnarray}\label{eq:lequ}
    \L_{\varrho_{\E_{\vec{p}}}}=1-P_{\max},
\end{eqnarray}
with $P_{\max}=\max\{P_1,P_2,P_3\}$ and quantifies the amount of discord-like correlation preserved in a maximally entangled state evolved by a Pauli channel.
Above, we dropped $A$ in the subscript for brevity because of the symmetry that Bell diagonal states possess on quantum correlation with respect to subsystems $A$ and $B$, i.e., $ \L_{\varrho_{\E_{\vec{p}}}}= \L_{\varrho_{\E_{\vec{p}}}|A}= \L_{\varrho_{\E_{\vec{p}}}|B}$.

\subsection{Channel-measurement compatibility}
\label{subsection:Incompatibility}
The concept of compatibility of a channel and a quantum measurement is aimed at aggregating possible disturbances a given measurement may cause and possible measurements may be applied with the same disturbance. The formal definition of channel-measurement compatibility is brought in the next definition.
\begin{definition}[Channel-measurement compatibility \cite{Teiko2018}]
\label{def:compatibility}
    A channel $\E$ and a meausrement $\M=\{M_i\}$ are compatible if there exist completely positive trace-nonincreasing operations $\Phi_i$ such that for any $i$ and any input state $\rho$
    \begin{equation}
        \E(\rho)=\sum\Phi_i(\rho)\quad\mathrm{and}\quad\Tr[M_i\rho]=\Tr[\Phi_i(\rho)].
    \end{equation}
\end{definition}

It has been shown in Ref. \cite{Teiko2018} that the measurement  $\MM$ is compatible with the  Pauli channel $\E_{\vec{p}}$ if and only if it holds that
\begin{eqnarray}
\label{eqn:Incompatibilityinequality}
    \frac{s^2 n_1^2}{{P}_1^2}+\frac{s^2 n_2^2}{{P}_2^2}+\frac{s^2 n_3^2}{{P}_3^2}\leq 1,
\end{eqnarray}
where $P_i$'s are given in Eq.~\eqref{Pi}. For this inequality to hold, if ${P}_i= 0$, its corresponding term must necessarily vanish.

For a given   $\MM$, we denote by $\PP\subset\CC$ the subset of the Pauli channels in the set of its compatible channels. Similarly,  we  refer to $\D_{\E_{\vec{p}}}$ as the set of unbiased binary qubit measurements compatible with  a given Pauli channel $\E_{\vec{p}}$. For future reference, we mention that $\UU$ and $\DD$ denote the compatible set of unital qubit channels with a given $\MM$ and the set of  measurements $\MM$ compatible with a unital qubit channel $\E_u$, respectively.

An immediate consequence of the inequality \eqref{eqn:Incompatibilityinequality}  is that if it holds for a POVM  with orientation  $\hat{\boldsymbol{n}}$ and  sharpness $s$, it also holds for all POVMs having the same direction but smaller sharpness  \cite{Teiko2018}.
Remarkably,  the set of channels compatible with a sharp POVM in the $\hat{\boldsymbol{n}}$-direction is also compatible with a POVM  having  the same direction, but arbitrary sharpness. That is to say,  $\P_{\!\scalemath{0.6}{\mathscr M}_{1, \hat{\boldsymbol{n}}}}\subseteq\PP\subseteq\P_{\!\scalemath{0.6}{\mathscr M}_{\!s', \hat{\boldsymbol{n}}}}$ for any $s^\prime \le s$.

By applying the Lagrange multiplier method, one can show ${P}_i\le 1$ for any $i$ where  the equality holds if  and only if  $p_0=p_i$ and $p_j=p_k$. Moreover, for  $P_{\max}=\max \{P_1,P_2,P_3\}$, the  compatibility inequality \eqref{eqn:Incompatibilityinequality}  gives rise to
\begin{eqnarray}\label{s<Pmax}
s\le {P}_{\max}\leq 1.
\end{eqnarray}
Both bounds of the above inequality are tight, meaning that in the set of compatible Pauli channels for any measurement $\MM$, there always exists a channel satisfying the right inequality and a channel satisfying the left one. The maximally depolarizing channel, located at the center of the tetrahedron of Pauli channels with $p_j=1/4$ for all $j$, satisfies the right inequality. As for the left inequality, there are four channels in $\PP$ introduced in the following remark.

\begin{remark}\label{Remark1}
For any measurement $\MM$, with its effects given by Eq. \eqref{eqn: POVM elements}, there exist at least four compatible channels $\E_{\vec{p}}$ on the boundary of $\PP$ where ${P}_i$, mentioned in Eq.~\eqref{Pi}, is equal to $s$ for any $i$. These channels are characterised by a probability vector $\vec{p}$ with $p_0=1-3p_1$, $p_1=p_2=p_3=\frac{1}{8}(1+s-\sqrt{1+2s-3s^2})$, and three different permutations of this vector.
\end{remark}

On the other hand, for a fixed Pauli channel $\E_{\vec{p}}$, that is for a fixed $P_{\max}$, there always exists a compatible measurement in $\D_{\E_{\vec{p}}}$ which satisfies with equality the left-hand side of the Eq.~\eqref{s<Pmax}. Trivially, this is the measurement with sharpness equal to $P_{\max}$ along the $i$-th principal axis determined by the $i$ that maximises $P_i$.

\begin{figure}[t]
\centering
\includegraphics[width=\columnwidth]{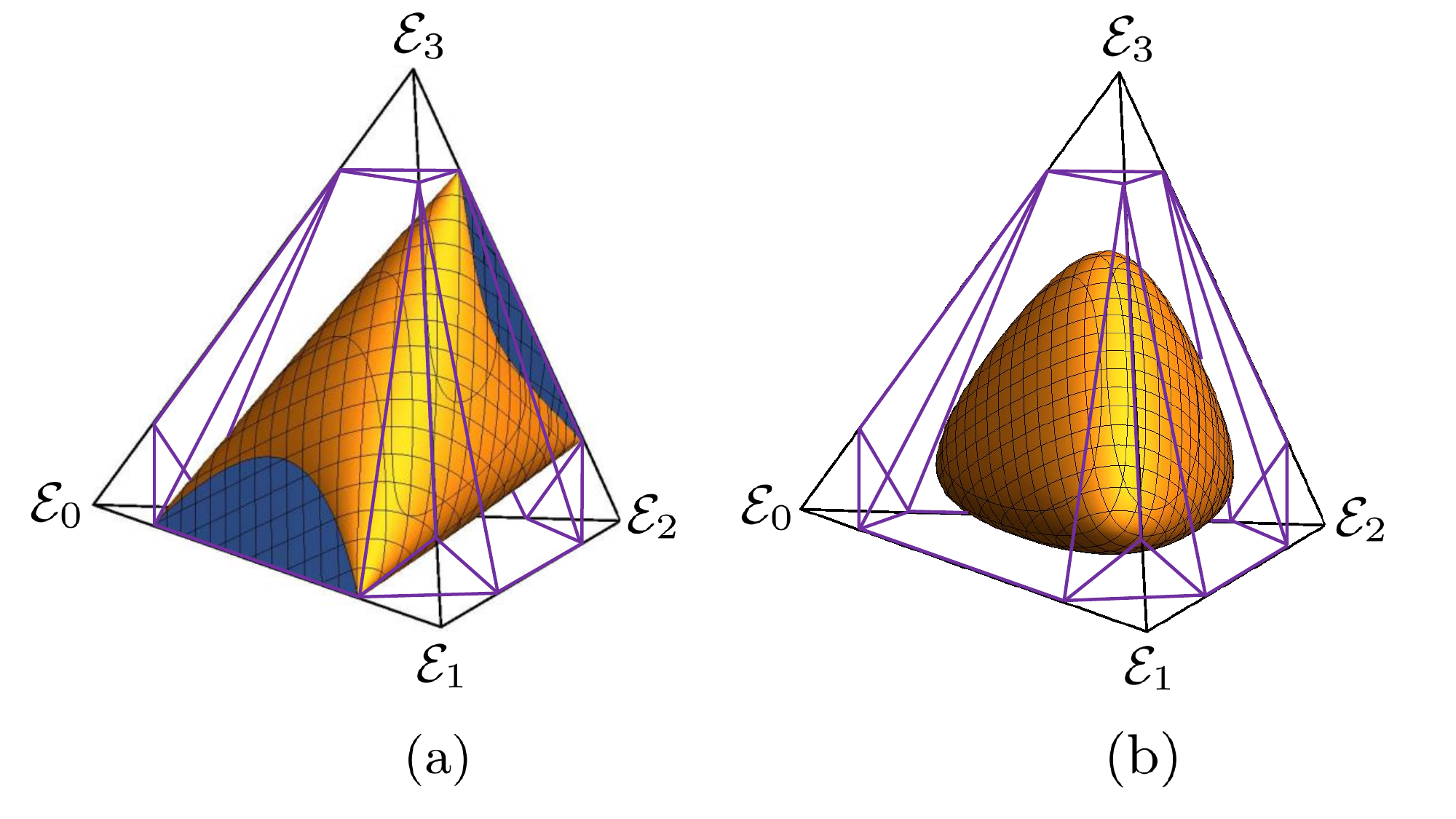}
    \caption{The tetrahedron of Pauli channels and the subset of compatible channels with a measurement $\MM$ where  $s=0.85$ and (a)   $\hat{\boldsymbol{n}}=\hat{\boldsymbol{x}}=(1,0,0)$, (b) $\hat{\boldsymbol{n}}=\frac{1}{\sqrt{3}}(1,1,1)$. Here, $\E_{i} (\rho)=\sigma_i\rho\sigma_i$  for $i=0,1,2,3$ is a unitary Pauli channel. The compatibility polytope $\Delta_s$ introduced in Definition \ref{def:polytope} is also presented and, as proved in Lemma~\ref{Lemma-Polytope}, it is a superset of the compatible set of Pauli channels.}
\label{channel}
\end{figure}
For a given measurement $\MM$ the set $\PP$, obtained through Eq. \eqref{eqn:Incompatibilityinequality},   is   represented  by the intersection of a deformed tetrahedron and the tetrahedron of Pauli channels, see Fig.~\ref{channel}. In specific cases, this set collapses to a line or a point. To see this, note that from
$1/{P}_i^2\geq 1$ we get   $\sum_{i=1}^3 n_i^2/{P}_i^2\geq 1$. Accordingly, in the particular case of a sharp POVM, i.e.,    $s=1$,   the inequality  \eqref{eqn:Incompatibilityinequality} reduces to the following equality
\begin{equation}\label{eqn:IncompatibilityinequalityPVM}
    \frac{n_1^2}{{P}_1^2}+\frac{n_2^2}{{P}_2^2}+\frac{n_3^2}{{P}_3^2}=1.
\end{equation}
The above shows the convex combination of some numbers $1/P^2_i\geq1$ is equal to one. This holds only if,  for each nonzero component $n_i$ of a sharp POVM, we have $P_i=1$. Now, if a sharp measurement  is directed along a principal axis, say $\hat{\boldsymbol{x}}=(1,0,0)$, then ${P}_1=1$. According to the above discussion, it holds for all the channels with $p_0=p_1$ and $p_2=p_3$ where $p_0+p_2=1/2$. This geometrically leads to a line segment across two opposite edges of the Pauli tetrahedron. On the other hand, when $\hat{\boldsymbol{n}}$ has at least two nonzero components, two of the $P_i$'s are equal to one. This enforces the third one to be equal to one as well. Therefore,  Eq. \eqref{eqn:IncompatibilityinequalityPVM}  has a single solution $P_i=1$ for $i=1,2,3$, occurring at  the center of the tetrahedron, i.e.,  $p_0=p_1=p_2=p_3=1/4$.

Alternatively, the set $\D_{\E_{\vec p}}$ of  POVMs compatible with a given channel $\E_{\vec{p}}$ is generally described as an ellipsoid inside the Bloch ball with semi-axes ${P}_1$, ${P}_2$, and ${P}_3$, see Eq.~\eqref{eqn:Incompatibilityinequality}. The ellipsoid coincides with the whole ball if and only if  $P_i=1$ for all $i$, or equivalently   $p_j=1/4$ for all  $j$. It collapses to the center of the Bloch ball if and only if $P_i=0$ for all $i$, or equivalently $p_j=\delta_{jk}$ for some $k$, i.e., when the Pauli channel becomes a unitary evolution.

\section{Results and discussion}
\label{section:results}
In this section,  we will present  our results on the following two questions; (i) what is the minimum possible disturbance a given measurement with sharpness $s$ may induce?, and (ii) given a fixed disturbance in postmeasurement state space, what is the sharpest measurement one might have applied?, i.e., how close could one have come to achieving a projective measurement?
Before presenting our main points, we note that to use the generalisation of channel-measurement compatibility from the set of Pauli channels, $\P$, to the set of unital qubit ones, $\U$, mentioned in Ref. \cite{Teiko2018}, we first need an explicit presentation of the generalisation in order to become suitable for our purpose. Therefore,  we will apply the following lemma from Ref.~\cite{Teiko2018}.
\begin{lemma}[Proposition 1 of Ref. \cite{Teiko2018}]
\label{lem: unitary-compatibility}
A quantum channel $\E$ is compatible with the measurement $\M$ if and only if $\V_2\circ\E\circ\V_1$, for two unitary channels $\V_j(\rho)=V_j\rho V_j^\dagger$ ($j=1,2$), is compatible with the POVM $V_1^\dagger\M V_1:=\{V_1^\dagger M_i V_1\}$.
\end{lemma}
Now to generalise, recall that any unital qubit channel $\E_u$ can be decomposed into  $\E_u=\V_2\circ\E_{\vec{p}}\circ\V_1$ for two  unitary channels $\V_1$ and $\V_2$. We  emphasise this decomposition is not unique. The freedom, however, at the level of Pauli channels allows only for the permutations of $\vec{p}$ and not any other change. Thus, it is enough to find one such decomposition for a unital channel $\E_u$. Applying the above lemma, to determine whether $\E_u$ and $\MM$ are compatible, the compatibility of the corresponding $\E_{\vec{p}}$ and $\M_{\!s, \hat{\boldsymbol{n}}'}=V_1\MM V_1^\dagger$ should be checked. This induces  on $\hat{\boldsymbol{n}}$ a rotation $R_{\scalemath{0.6}{V_1}}$ corresponding to the unitary operator $V_1$ through $\mathrm{SU}(2)$-$\mathrm{SO}(3)$ homomorphism, i.e.,
\begin{equation}
\label{eq:n prime}
    \hat{\boldsymbol{n}}'=R_{\scalemath{0.6}{V_1}}\hat{\boldsymbol{n}}.
\end{equation}

Applying compatibility inequality \eqref{eqn:Incompatibilityinequality} for the measurement $\M_{\!s, \hat{\boldsymbol{n}}' }$, we find the set of Pauli channels compatible with this measurement denoted by $\P_{\!\scalemath{0.6}{\mathscr M}_{\!s, \hat{\boldsymbol{n}}'}}$, then by using the above lemma we get the set of all unital qubit channels compatible with a measurement $\MM$ as
\begin{equation}
\label{eq:compatible-unital}
    \UU= \{\V_2\circ\P_{\!\scalemath{0.6}{\mathscr M}_{\!s, \hat{\boldsymbol{n}}'}}\circ\V_1: \ \forall \ \V_1, \V_2 \in \textrm{unitary\ channels}\} ,
\end{equation}
where $\hat{\boldsymbol{n}}'$ is given in Eq.~\eqref{eq:n prime}. By concatenation of a channel and a set, we mean the concatenation of the channel and all members of the set.

Alternatively, among the unbiased binary qubit measurements $\MM$, those compatible with a unital qubit channel $\E_u$ belong to
\begin{equation}
\label{eq:unital compatible set}
    \DD=V_1^\dagger \D_{\E_{\vec{p}}} V_1,
\end{equation}
where, by multiplication of a matrix and a set we mean multiplication of the matrix with each element of the set. Eq.~\eqref{eq:unital compatible set} gives a rotation, by $R_{V_1^\dagger}$, of the ellipsoid of the compatible measurements $\D_{\E_{\vec{p}}}$ in the Bloch ball, see Section~\ref{subsection:Incompatibility}.

\subsection{Minimum disturbance caused by a fixed POVM}
\label{subsection: minimum disturbance}
In this subsection, we study  the extent to which a fixed measurement $\MM$ can maintain the state space and its quantum resources, specifically coherence and discord-like correlations. To present our results, we make use of the following definition and technical lemma.
\begin{definition}[Compatibility polytope $\Delta_s$]
\label{def:polytope}
    For any $s\in[0,1]$, the compatibility polytope, denoted by $\Delta_s$, is the convex hull of Pauli channels corresponding to the probability vectors given by
    \begin{subequations}\label{Pauli-edge-s}
        \begin{align}
            \vec{q}_1&=(\frac{1+\sqrt{1-s^2}}{2},\frac{1- \sqrt{1-s^2}}{2},0,0),\label{subeq:q1}\\
            \vec{q}_2&=(\frac{1+\sqrt{1-s^2}}{2},0,\frac{1- \sqrt{1-s^2}}{2},0),\\
            \vec{q}_3&=(\frac{1+\sqrt{1-s^2}}{2},0,0,\frac{1- \sqrt{1-s^2}}{2}),
        \end{align}
    \end{subequations}
    and their other nine different permutations.
\end{definition}
 Geometrically, $\Delta_s$ is obtained by truncating all four vertices of the regular tetrahedron of Pauli channels, see Fig.~\ref{channel}. For $s=0$,  $\Delta_0$ coincides with the Pauli tetrahedron, however,  for $s=1$, $\Delta_1$ reduces to a regular octahedron. We use this definition in the following lemma the proof of which is brought in Appendix~\ref{app:lemma-polytope}.

\begin{lemma}\label{Lemma-Polytope}
The set $\PP$ of  compatible Pauli channels with a given measurement $\MM$ is  always included in the compatibility polytope $\Delta_s$. Moreover, the edges of $\Delta_s$ can be touched by the compatible set of channels if and only if the corresponding POVM is directed along a principal axis.
\end{lemma}

As previously mentioned, different methods exist to measure a given POVM, which results in different postmeasurement state spaces. Depending on the specific operational objective one has for the postmeasurement states, the best method for the measurement can be defined.  A natural choice, however, is to have the postmeasurement states as similar as possible to the premeasurement ones. This similarity can be quantified by employing the notion of average fidelity between pre- and postmeasurement states. Since the compatible channels $\E_\M$ model the postmeasurement state spaces for a given measurement $\M$, the best similarity, in the above sense, is then equivalent to
\begin{equation}
\label{eq:measurement-fidelity}
        F_{\!\C}(\M):=\max_{\C_{\M}}\overline{F}_{\E_{\M}},
\end{equation}
where $\overline{F}_{\E_{\M}}$ is defined in Eq.~\eqref{eq:fidelity}. Hereafter, we refer to this quantity as \emph{the best fidelity over $\C_\M$}. Obviously, once $\C_\M$ is the whole set of compatible channels, the above gives the best fidelity of $\M$ overall.
For the unital channels considered in this paper, we can define \emph{the best fidelity over ${\UU}$}  by
\begin{align}
\label{eq:measurement-fidelity-unital}
F_{\!\U}(\MM):=\max_{\UU}\overline{F}_\E.
\end{align}

Now, we present our first main result on a tradeoff relation between the sharpness-disturbance in the next theorem. Informally, it shows that the sharper a given POVM is, the less similar pre- and postmeasurement states, in the best-case scenario, can be.

\begin{theorem}
\label{thm: measurement fidelity-sharpness}
   Assume a measurement $\MM$ and its all compatible unital qubit channels. Then,  the best fidelity over $\UU$,  Eq.~\eqref{eq:measurement-fidelity-unital}, satisfies the following tradeoff relation with sharpness
   \begin{equation}
   \label{eq:unital-fidelity-sharpness}
       (3F_{\!\U}(\MM)-2)^2+s^2=1.
   \end{equation}
\end{theorem}

Before proceeding with the proof, let us examine the theorem for two extreme cases of a projective and the maximally noisy POVMs. The first implies that the best measurement strategy, as far as being restricted to the set of unital channels, definitely disturbs the states because $F_{\!\U}(\M_{1,\hat{\boldsymbol{n}}})=2/3$ and cannot achieve its maximum value of unity equal to one. This amount of fidelity corresponds to the second method mentioned in the caption of Fig.~\ref{fig:indirect-measurement}. On the other hand, $F_{\!\U}(\M_{0,\hat{\boldsymbol{n}}})=1$, which is a trivial bound showing it can be measured without any disturbance.
\begin{proof}
    To prove  the theorem,  by applying Eqs.~\eqref{eq:n prime} and \eqref{eq:compatible-unital}, we get
    \begin{align}
        F_{\!\U}(\MM)&=\max_{\UU}\overline{F}_\E=\max_{\V_2,\V_1}\max_{\P_{\!\scalemath{0.6}{\mathscr M}_{\!s, \hat{\boldsymbol{n}}'}}}\overline{F}_{(\V_2\circ\E_{\vec{p}}\circ\V_1)}\nonumber\\
        &=\max_{\hat{\boldsymbol{m}}}\max_{\P_{\!\scalemath{0.6}{\mathscr M}_{\!s, \hat{\boldsymbol{m}}}}}\overline{\mathscr{F}}_{\E_{\vec{p}}}=\max_{\hat{\boldsymbol{m}}}\max_{\P_{\!\scalemath{0.6}{\mathscr M}_{\!s, \hat{\boldsymbol{m}}}}}\frac13(1+2p_{m})\nonumber\\
        &=\frac{1}{3}\left(2+\sqrt{1-s^2}\right).
    \end{align}
    Above, the first equality in the second line is due to Eq.~\eqref{eq:corrected-fidelity} and \eqref{eq:corrected-fidelity-pauli}, which implies that for a Pauli channel the corrected fidelity is achievable by a unitary map. The subsequent equality is also a result of Lemma~\ref{Lemma-Polytope} implying one can always find a direction to get, and never exceed,  $p_m=(1+\sqrt{1-s^2})/2$ for a compatible Pauli channel.
\end{proof}

While the above theorem shows that the disturbance due to measuring an informative POVM is inevitable, for practical reasons, one may wonder how quantum resources are affected by such a POVM. Strictly speaking, an alternative approach to defining the best measurement strategy can be the one that preserves more quantum resources in the state space. Here, we study two important characteristics  of such quantum resources, namely quantum coherence and quantum correlation once the measurement is applied locally. We consider for the former the average coherence of the postmeasurement state space and for the latter the LQU of a maximally entangled state measured locally by $\M$. The best measuring strategy then is the one that maximises one of these two, depending on the practical intentions. Following the discussion before Theorem~\ref{thm: measurement fidelity-sharpness}, for a measurement $\M$, these amounts are respectively given by
\begin{align}
    \Q_\C(\M)&:=\max_{\C_\M}\Q_{\E_\M},\label{eq:measurement-quantumness}\\
    \L_\C(\M)&:=\max_{\C_\M}\L_{\varrho_{\E_\M}|A},\label{eq:measurement-LQU}
\end{align}
where $\Q_{\E_\M}$ is the quantumness \eqref{eq:quantumness} of the channel $\E_\M$ compatible with $\M$ which quantifies the average coherence preserved in the postmeasurement state space. Also, $\varrho_{\E_\M}$ is the Choi-Jamio{\l}kowski state  of the compatible channel $\E_\M$ whose LQU,  represented by $\L_{\varrho_{\E}|A}$, quantifies the discord-like correlations  preserved in the maximally entangled state after being measured locally by $\M$. Henceforward, we call the first \emph{the best  quantumness over $\C_\M$} and the latter  \emph{the best LQU over $\C_\M$}.
Restricting to the set of unital channels,  we  define \emph{the best quantumness and the best LQU  over ${\UU}$}, respectively,   by
\begin{align}
    \Q_{\U}(\MM)&:=\max_{\UU}\Q_\E, \label{eq:measurement-quantumness-unital}\\
    \L_{\U}(\MM)&:=\max_{\UU}\L_{\varrho_{\E}}.
    \label{eq:measurement-LQU-unital}
\end{align}
Similar to Theorem~\ref{thm: measurement fidelity-sharpness}, the next theorem presents a tradeoff relation between the sharpness of a given $\MM$ and its best quantumness and LQU over $\UU$.

\begin{theorem}
\label{thm:quantum-resources-sharpness}
    Consider an unbiased binary qubit measurement $\MM$ and  restricted its  set of compatible channels to the unital qubit ones, the best  quantumness and the best LQU over $\UU$, defined by Eqs.~\eqref{eq:measurement-quantumness-unital} and \eqref{eq:measurement-LQU-unital},  satisfy
    \begin{align}
    \Q_{\U}(\MM)+s^2&= 1,\label{eq: measurement quantumnes-sharpness}\\
    \L_{\U}(\MM)+s & =1.\label{eq: measurement lqu-sharpness}
    \end{align}
\end{theorem}
We bring the proof of this theorem in Appendix~\ref{app:quantum-resources-sharpness}. In  this appendix, we also show that the equality in   (\ref{eq: measurement quantumnes-sharpness})  reduces to the inequality $\le 1$   if  the best quantumness is defined over the Pauli channels instead of the unital ones.
An exception however exists, namely, for three  measurements for which $\hat{\boldsymbol{n}}$ is directed along one of the  principal axes $\hat{\boldsymbol{x}}$, $\hat{\boldsymbol{y}}$, and $\hat{\boldsymbol{z}}$. In these particular  cases,  the equality \eqref{eq: measurement quantumnes-sharpness} is still satisfied even for the best quantumness  over the Pauli channels.

This theorem significantly proves for a projective measurement, i.e., $s=1$,  the postmeasurement state space, independent of measuring strategy, becomes entirely classical. That means all postmeasurement states become compatible.  On the other hand,  any informative measurement does destroy quantum resources to a degree determined by the sharpness of the measurement.
\subsection{The sharpest measurement for a fixed disturbance}
\label{subsection: sharpest measurement}
In this subsection, we delve into the inverse problem. Here, we begin  with a predetermined level of disturbance and examine among the measurements with this specific postmeasurement state space, what is their maximum sharpness. Having the set of measurements compatible with a given quantum channel, the above is equivalent to starting with a quantum channel and then finding the maximum sharpness a measurement in its compatible set possesses. We will bring some tradeoff relations in what follows proving the more disturbance exists in the postmeasurement space, the sharper the measurement might have been.

Let $\S_\M$ define a measure for the sharpness of a given quantum measurement $\M$. Following the terminology method employed in the previous subsection, we introduce the notion of \emph{the best sharpness over $\D_\E$} as the maximum sharpness of a measurement $\M\in\D_{\E}$ compatible with $\E$, i.e,
\begin{equation}
\label{eq:channel sharpness}
    \S_{\D}(\E):=\max_{\D_{\E}} \S_\M.
\end{equation}

Various measures of sharpness can be employed, however, in the specific case of an unbiased binary qubit measurement that we are considering, a natural choice is $s$, as we discussed after Eq.\eqref{eqn: POVM elements}. Recalling Eq.\eqref{s<Pmax} and the discussion after Remark~\ref{Remark1}, for a Pauli channel the best sharpness of $\E_{\vec{p}}$ over $\D_{\E_{\vec{p}}}$ is precisely equal to $P_{\max}$,
\begin{equation}
\label{eq:pauli sharpness s}
    \S_{\D}(\E_{\vec{p}}):=\max_{\D_{\E_{\vec{p}}}}\ s=P_{\max}.
\end{equation}
\\
As mentioned at the beginning of this section, the set $\DD$ is obtained from $\D_{\E_{\vec{p}}}$ by some rotation \eqref{eq:unital compatible set}. This does not change the maximum sharpness in the set of compatible measurements. Thus, we have
\begin{equation}
\label{eq: unital sharpness}
     \S_{\D}(\E_u):=\max_{\D_{\E_u}}\ s=P_{\max}.
\end{equation}

In what follows, we discuss  how $ \S_{\D}(\E_{u})$ is related to the input-output fidelity that measures the disturbance of the postmeasurement space.

The next theorem presents a tradeoff relation between the sharpness of a unital qubit channel with all positive signed-singular-values \cite{ZyczkowskiBook2017} and its input-output fidelity. It proves for a given postmeasurement state space,  the greater the fidelity between the premeasurement and postmeasurement state spaces, the lower the sharpness of the measurement that maps them, in the optimal scenario.

\begin{theorem}
\label{thm:fidelity-channel sharpness}
    Assume a unital qubit channel $\E_u=\V_2\circ\E_{\vec{p}}\circ\V_1$ corresponding to a Pauli channel $\E_{\vec{p}}$. Also, let the largest component of $\vec{p}$ exceed $1/2$, i.e., $p_{m}\geq1/2$. Assigning the image of $\E_{u}$ to the outputs of a measurement $\MM$, then the following tradeoff relation applies to the corrected fidelity  $ \overline{\mathscr{F}}_{\E_{u}}$ given by \eqref{eq:corrected-fidelity},  and the best sharpness over $\D_{\E_{u}}$ defined  by Eq.~\eqref{eq: unital sharpness}
    \begin{equation}
    \label{eq:fidelity-channel sharpness}
        (3\overline{\mathscr{F}}_{\E_{u}}-2)^2+(\S_{\D}(\E_{u}))^2\leq 1.
    \end{equation}
 \end{theorem}

\begin{proof}
    For $\vec{p}^{\downarrow}$ being $\vec{p}$ rearranged in a non-increasing order, $P_{\max}$ is given by Eq.~\eqref{app-eq:P-max}.
    Hence, for any map $\E_{\vec{p}}$ appearing in the decomposition of $\E_u$, independent of the decomposition, $P_{\max}$ is fixed. The upper-bound for the channels with $p_0^{\downarrow}=p_m\geq1/2$ is obtained by
    \begin{align}
    S_{\D}(\E_{\vec{p}})&=P_{\max}=2\left(\sqrt{p_0^{\downarrow}p_1^{\downarrow}}+\sqrt{p_2^{\downarrow}p_3^{\downarrow}}\right)\nonumber\\
        &\leq 2\left(\sqrt{p_0^{\downarrow}p_1^{\downarrow}}+\frac{1-p_0^{\downarrow}-p_1^{\downarrow}}{2}\right) \nonumber\\
        &\leq2\sqrt{p_0^{\downarrow}(1-p_0^{\downarrow})}.
    \end{align}
    Here, by optimising the term in the second line with respect to $p_1^{\downarrow}$ and considering the assumption of  $p_0^{\downarrow}\geq1/2$, the last inequality is obtained for $p_1^{\downarrow}=1-p_0^{\downarrow}$. Replacing $p_0^{\downarrow}=p_m$ from Eq.~\eqref{eq:corrected-fidelity-pauli} will complete the proof.
\end{proof}

The term $(3\overline{\mathscr{F}}_{\E_{u}}-2)^2$ in Eq.~\eqref{eq:fidelity-channel sharpness} for $1/2\leq p_m\leq1$, i.e., for $2/3\leq\overline{\mathscr{F}}_{\E_{u}} \leq1$, monotonically increases from $0$ to $1$. This theorem shows at a cost of certain disturbance in the postmeasurement space, how sharp the measurement might have been.

The technical reason why the above theorem is restricted to the channels with $p_m\geq1/2$ is obvious; while the best sharpness over $\D_{\E{u}}$, determined by $P_{\max}$, depends on $\vec{p}$ totally, the corrected fidelity relies only on its largest component. For $p_m\leq1/2$, it is always possible to find a channel with $P_{\max}=1$. An example of which is the channel that corresponds to $p_m=p_0=p_1=\frac{1}{2}-x$ and $p_2=p_3=x$ where $x\leq1/4$ is the deviation of $p_0$ from $1/2$. Thus, it is unexpected to formulate such a tradeoff relation for $p_m<1/2$.

To address this issue and take other notions of postmeasurement disturbance in terms of states' resourcefulness into account, we bring in the next theorem two tradeoff relations between the quantumness and LQU of a given $\E_{u}$ and its best sharpness over $\D_{\E_{u}}$.

\begin{theorem}
\label{thm:quantumness-channel sharpness}
    Let $\E_{u}$ be a unital qubit channel and restricted its set of compatible measurements to the unbiased binary qubit ones. Then, the best sharpness over $\D_{\E_{u}}$, Eq.~\eqref{eq: unital sharpness}, respects
    \begin{align}
       \Q_{\E_{u}}+(\S_\D(\E_{u}))^2&\leq 1,\label{eq:quantumness-channel sharppness}\\
       \L_{\varrho_{\E_{u}}}+\S_\D(\E_{u})&=1,\label{eq:lqu-channel sharppness}
    \end{align}
    where $\Q_{\E_{u}}$ and $\L_{\varrho_{\E_{u}}}$ are  the quantumness and LQU of $\E_{u}$, respectively, defined in Section~\ref{subsection:channels}.
\end{theorem}

\begin{proof}
    We begin the proof by noticing the decomposition of $\E_u=\V_2\circ\E_{\vec{p}}\circ\V_1$ and the fact that both $\Q_\E$ and $\L_{\varrho_\E}$ are invariant under such unitary evolutions, see \eqref{eq:quantumness-invariance} and \eqref{eq:lqu-invariance}. Therefore, the proof of Eq.~\eqref{eq:quantumness-channel sharppness} is a straight consequence of the fact that $\Q_{\E_u}=\Q_{\E_{\vec{p}}}$, Eq.~\eqref{eq: unital sharpness}, and $1-P_{\max}^2\geq \Q_{\E_{\vec{p}}}$ as proved in Eq.~\eqref{app-eq:P-max}, Appendix~\ref{app:quantum-resources-sharpness}.
\\
    On the other hand, the proof of Eq.~\eqref{eq:lqu-channel sharppness} is an immediate result of the equality $\L_{\varrho_{\E_u}}=\L_{\varrho_{\vec{p}}}$, Eq.~\eqref{eq:lequ}, and Eq.~\eqref{eq: unital sharpness}, which completes the poof of the theorem.
\end{proof}

This theorem proves if a POVM destroys a certain amount of quantum resources in the state space, it cannot possess an arbitrary sharpness.
From both Theorems~\ref{thm:quantum-resources-sharpness} and \ref{thm:quantumness-channel sharpness}, it is now obvious that a POVM of the form $\MM$ is sharp  if and only if its compatible set of  unital channels  includes  only q-c   channels, i.e., the post-measurement states are classical with respect to some basis.  This implies that for any unsharp measurements, one can find some measurement strategies to maintain some quantum resources in the state space.

\section{Conclusion}
\label{sec:Conclusion}
In this paper, we have considered the question that how gaining information by a measurement can affect the states of the system and their quantum resourcefulness. To address this question, we have used channel-measurement compatibility as  a framework that characterises all possible disturbances  that might have been caused by a given measurement. The framework allows also for identifying all possible measurements inducing a given disturbance.

Restricted to unbiased binary qubit measurements and unital qubit channels, we proved in Theorem~\ref{thm: measurement fidelity-sharpness}, that any informative POVM inevitably changes the states. Whereas in Theorem~\ref{thm:quantum-resources-sharpness}, we showed this change is necessarily destructive in terms of quantum resources measured by coherence and discord-like correlations. That is to say, the amount of coherence and discord-like correlations, in case the measurement is local, in the postmeasurement states are less than the premeasurement states to the extent determined by the measurement sharpness.  In the extreme case of a projective measurement, which represents the most informative measurements, all these quantum resources are entirely erased.

The study of the  inverse problem, i.e., the question that how informative the measurement might have been for a certain disturbance, results in Theorem~\ref{thm:fidelity-channel sharpness} which proves that the best informativity is bounded by the amount of the similarity of the pre- and postmeasurement states. Theorem~\ref{thm:quantumness-channel sharpness}, however, presents the tradeoff relation between the best informativity and quantum resources preserved in the postmeasurement state space.

An obvious approach for further work is to extend the results to a larger class of measurements and postmeasurement models rather than unbiased binary $\MM$ and unital qubit channels $\U$. For such a generalisation, an immediate candidate can be  the set of binary qubit measurements that are not necessarily unbiased, i.e., $\M=(M_1,M_2)$ with $\Tr M_1\neq\Tr M_2$. In this case, we expect to have tradeoff relations based on the disturbance and the measurement sharpness and biasedness. This is because the extreme case of a projective measurement $\M=(\Id,0)$ which is trivially sharp~\cite{BuschFP2009} is not informative and thus destructive. So the POVM sharpness is not the only effective factor.

\section*{acknowledgment}
This work was supported by Ferdowsi University of Mashhad under Grant No. 3/55453 (1400/10/19).
FS acknowledges support from the Foundation for Polish Science through TEAM-NET project (contract no. POIR.04.04.00-00-17C1/18-00).

\appendix

\section{Sharpness Measures}
\label{sharpness-measures}
Considering the  unsharpness measures defined previously, for instance, in \cite{LiuPRA2021} and \cite{MitraIJTF2022}, we will show that  for the unbiased binary qubit measurements both quantifiers reduce to $(1-s^2)$, up to a normalization  factor. Accordingly, the related sharpness measures are  given simply by $s^2$.  As $s^2$ is monotonically increasing function of $s$, we can take   $s$ also as a measure of sharpness.
\\
For a general measurement $\M=\{M_i \}_{i=0}^{N-1}$, the  unsharpness quantifier  based on the  uncertainty  is defined    by   \cite{LiuPRA2021}
\begin{align*}
f(\M)=\|F(\M)\|_1,
\end{align*}
where $\|X\|_1=\sum_{i,j} |X_{ij}|$ denotes the $l^1$-norm of the matrix $X=[X_{ij}]$, and   $F(\M)=[r_{ij}]$  with $r_{ij}=\frac{1}{d}[\delta_{ij} \Tr(M_i) - \Tr(M_i M_j)]$.
For an unbiased binary qubit measurement
$\MM=\{M_+,M_-\}$, we get
\begin{align*}
	f(\MM)=\|F(\MM)\|_1 = \sum_{i,j=1}^2 |r_{ij}|=1-s^2,
\end{align*}
as such, $1-f(\MM)=s^2$ gives the sharpness of the unbiased binary qubit measurements.
\\
Another unsharpness measure based on Luder's instrument is defined by \cite{MitraIJTF2022}
\begin{align*}
	\E^L(\M)=\| \Id - \sum_i M_i^2\|,
\end{align*}
where $\|.\|$ denotes the operator norm. For an unbiased binary qubit measurement $\MM$ we obtain
\begin{align*}
	\E^L(\MM)=\frac{1}{2} (1-s^2),
\end{align*}
which is, up to a normalisation factor, the same as the above measure so leads to the same sharpness $s^2$ for the unbiased binary qubit measurements.


\section{Proof of Lemma~\ref{Lemma-Polytope}}
\label{app:lemma-polytope}
To prove the first part of the lemma, we will show $\PP$ never crosses the boundaries of $\Delta_s$. To see that, consider  a channel defined by a convex combination of three vertices mentioned in Eq.~\eqref{Pauli-edge-s}, say
\begin{equation}
\vec{p}=w_1\vec{q}_1+w_2\vec{q}_2+w_3\vec{q}_3,
\end{equation}
with non-negative $w_i$ satisfying $\sum w_i=1$. For different choices of probabilities $w_i$,  these vectors construct a boundary of $\Delta_s$ that is not in common with boundaries of the tetrahedron, see Fig.~\ref{channel}.  We will show the Pauli channel $\E_{\vec{p}}$ corresponding to the probability  $\vec{p}$ mentioned above can be compatible with $\MM$ only in an extreme case.  For this channel, we get from Eq.~\eqref{Pi}
\begin{equation}
    P_i=s\sqrt{w_i}+(1-\sqrt{1-s^2})\sqrt{w_jw_k}.
\end{equation}
This, however, implies
\begin{equation}
    P_{\max}\leq s\sqrt{w_{\max}}+(1-\sqrt{1-s^2})(\frac{1-w_{\max}}{2}),
\end{equation}
which in turn means $P_{\max}\leq s$ with the equality possible only for the case with $w_{\max}=1$. This generally contradicts Eq. \eqref{s<Pmax} except at some vertices. The same  is true for the other three comparable boundaries of $\Delta_s$. This completes the proof of the first part.

To prove the second part,  suppose the POVM is directed along a principal axis, say $\hat{\boldsymbol{x}}=(1,0,0)$.  In this case, the inequality \eqref{eqn:Incompatibilityinequality} leads  to $s\le 2(\sqrt{p_0p_1}+\sqrt{p_2p_3})$.   For the edge defined by  $p_2=p_3=0$, the inequality reduces to $s\le 2\sqrt{p_0p_1}$ which  touches the line segment between the points $\vec{q}_1$, given by Eq.~\eqref{subeq:q1}, and its permutation $\vec{q'}_{\!\!1}=(\frac{1-\sqrt{1-s^2}}{2},\frac{1+ \sqrt{1-s^2}}{2},0,0),$ on the edge. The same occurs at the edge defined by $p_0=p_1=0$. A similar result holds for the measurements along the other two principal axes  $\hat{\boldsymbol{y}}$ and $\hat{\boldsymbol{z}}$.

To complete the proof, we show  if $\PP$ intersects with $\Delta_s$ at its edges, the corresponding measurement is necessarily directed  along a principal axis. Assume the case where both sets intersect at the edge defined by $p_2=p_3=0$. This implies $P_2=P_3=0$, accordingly, $n_2=n_3=0$. This means the measurement is along $\hat{\boldsymbol{x}}$. The same result holds for the intersection on other edges and other principal axes as well, which completes the proof.

\section{Proof of Theorem \ref{thm:quantum-resources-sharpness}}
\label{app:quantum-resources-sharpness}
We start the proof by proving Eq.~\eqref{eq: measurement quantumnes-sharpness}. To this end, we first show
 \begin{equation}
 \label{eq:quantumness-upperbound}
     \forall\hat{\boldsymbol{m}}:\quad\Q_{\P}(\mathscr{M}_{\!s, \hat{\boldsymbol{m}}}):=\max_{\P_{\!\scalemath{0.6}{\mathscr M}_{\!s, \hat{\boldsymbol{m}}}}}\ \Q(\E_{\vec{p}})\leq1-s^2.
\end{equation}
That is to say, the best quantumness over $\P_{\mathscr M_{\!s, \hat{\boldsymbol{m}}}}$ is upper-bounded with a tight bound achievable for some directions $\hat{\boldsymbol{m}}$. The proof of \eqref{eq: measurement quantumnes-sharpness} is then an immediate consequence of the invariance of $\Q(\E)$ under unitary channels, see Eqs.~\eqref{eq:quantumness-invariance} and \eqref{eq:compatible-unital}.

To prove Eq.~\eqref{eq:quantumness-upperbound}, let $\vec{p}^{\downarrow}$ denote $\vec{p}$ rearranged in a non-increasing order.  It is then easy to show among $P_i$'s given by Eq.~\eqref{Pi}, the largest one is
\begin{equation}
\label{app-eq:P-max}
    P_{\max}=2\left(\sqrt{p_0^{\downarrow}p_1^{\downarrow}}+\sqrt{p_2^{\downarrow}p_3^{\downarrow}}\right).
\end{equation}
Moreover, for any $\E_{\vec{p}}$ to belong to $\P_{\!\scalemath{0.6}{\mathscr M}_{\!s, \hat{\boldsymbol{m}}}}$,  a necessary condition  is to satisfy Eq. \eqref{s<Pmax}, which in turn implies
\begin{equation}
\label{aap-eq:s2-bound}
    1-s^2\geq1-P_{\max}^2=Q(\E_{\vec{p}})+\frac{1}{2}\left(T_2^2+T_3^2\right),
\end{equation}
where ${T}_{i}=2(\sqrt{p_0^{\downarrow} p_i^{\downarrow}}-\sqrt{p_j^{\downarrow} p_k^{\downarrow}})$ for $i,j,k$ being different choices of $(1,2,3)$. Since this has to hold for any compatible channel, it must hold for the compatible channel with maximum quantumness. This completes the proof of Eq.~\eqref{eq:quantumness-upperbound}. To see its tightness consider $\hat{\boldsymbol{m}}$ to align a principal axis. Being a tight bound is then a result of Lemma~\ref{Lemma-Polytope} implying for a POVM along a principal axis, four of the extreme points of $\Delta_s$ also belong to the set of compatible channels. While the bound is tight for all extreme points of $\Delta_s$.  Eq.~\eqref{eq: measurement quantumnes-sharpness} is now straightly obtained by the fact  that one can always find a unitary map $V_1$ to bring $\hat{\boldsymbol{n}}$ along a principal direction while $\Q$ is unitarily invariant \eqref{eq:quantumness-invariance}.

The proof of Eq.~\eqref{eq: measurement lqu-sharpness} is an immediate result of the invariance of LQU under unitary channels, Eq.~\eqref{eq:lqu-invariance}, noting that any unital channel is unitarily equivalent to a Pauli one, whereas, Eq.~\eqref{eq:lequ} gives the LQU for the latter. Restricted to the set of Pauli channels compatible with $\MM$,  Eq.~\eqref{s<Pmax} implies that the maximum achievable LQU for a compatible Pauli channel is $1-s$, and Remark~\ref{Remark1} guarantees that this bound is tight which completes the proof of Eq.~\eqref{eq: measurement lqu-sharpness} and Theorem~\ref{thm:quantum-resources-sharpness}.


\newpage
\bibliography{References}
\end{document}